  \def\acknowledgement{\par\addvspace{17pt}\small\rmfamily
  \trivlist\if!\ackname!\item[]\else
  \item[\hskip\labelsep
  {\bfseries\ackname}]\fi}
  
  \newenvironment{acknowledgements}{\begin{acknowledgement}}
  {\end{acknowledgement}}
  \def\trans@english{\switcht@albion}
  \def\switcht@albion{\def\ackname{Acknowledgements}%
  }\switcht@albion
\makeatother

\def\keywordname{{\bfseries Keywords}}%
\def\keywords#1{\par\addvspace\medskipamount{\rightskip=0pt plus1cm
\def\and{\ifhmode\unskip\nobreak\fi\ $\cdot$
}\keywordname\enspace\ignorespaces#1\par}}

\documentclass[11pt]{article}

\usepackage{mathrsfs}
\usepackage{float}
\usepackage{cite}
\usepackage[misc]{ifsym}
\usepackage{makecell}
\usepackage{slashbox}
\usepackage{amsmath} 
\usepackage{amsfonts} 
\usepackage{amssymb} 
\usepackage{amsthm}
\usepackage[english]{babel}
\usepackage{graphicx}   
\usepackage{pict2e}
\usepackage{url}       
\usepackage{bm}         
\usepackage[justification=centering]{caption}
\usepackage{multirow}
\usepackage{booktabs}
\usepackage{algorithm}
\usepackage{algorithmic}
\usepackage[left=3cm,right=3cm,top=2.50cm,bottom=2.50cm]{geometry} 
\usepackage[colorlinks,linkcolor=blue,anchorcolor=black,citecolor=blue]{hyperref}

\hypersetup{colorlinks,bookmarks,unicode}
\hypersetup{urlcolor=blue}
\def\email#1{\tt#1}

\newcommand{\tl}{\textnormal}
\newcommand{\rank}{\textrm{Rank}}
\newcommand{\supp}{\textnormal{Supp}_R}

\newcommand{\aut}{\tl{Aut}_{\mathbb{F}_{q^m}/\mathbb{F}_q}}

\newtheorem{theorem}{Theorem}
\newtheorem{corollary}{Corollary}

\newtheorem{proposition}{Proposition}

\theoremstyle{remark}
\newtheorem{definition}[theorem]{Definition}
\newtheorem{remark}{Remark}

\date{}

\begin{document}

\title{Semilinear Transformations in Coding Theory: A New Technique in Code-Based Cryptography}

\author{
Wenshuo Guo\textsuperscript{(\Letter)}\ and Fang-Wei Fu\\ \\
Chern Institute of Mathematics and LPMC, Nankai University, Tianjin 300071, China\\
\email{ws\_guo@mail.nankai.edu.cn, fwfu@nankai.edu.cn}
}

\maketitle

\begin{abstract}
This paper presents a new technique for disturbing the algebraic structure of linear codes in code-based cryptography. This is a new attempt to exploit Gabidulin codes in the McEliece setting and almost all the previous cryptosystems of this type have been completely or partially broken. To be specific, we introduce the so-called semilinear transformation in coding theory, which is defined through an $\mathbb{F}_q$-linear automorphism of $\mathbb{F}_{q^m}$, then apply them to construct a public key encryption scheme. Our analysis shows that this scheme can resist all the existing distinguisher attacks, such as Overbeck's attack and Coggia-Couvreur attack. Meanwhile, we endow the underlying Gabidulin code with the so-called partial cyclic structure to reduce the public key size. Compared with some other code-based cryptosystems, our proposal has a much more compact representation of public keys. For instance, 2592 bytes are enough for our proposal to achieve the security of 256 bits, almost 403 times smaller than that of Classic McEliece entering the third round of the NIST PQC project.
\keywords{Post-quantum cryptography \and Code-based cryptography \and Rank metric codes \and Gabidulin codes \and Partial cyclic codes \and Semilinear transformations}
\end{abstract}

\section{Introduction}
Over the past decades, post-quantum cryptosystems (PQCs) have been drawing more and more attention from the cryptographic community. The most important advantage of PQCs is their potential resistance against attacks from quantum computers. In post-quantum cryptography, cryptosystems based on coding theory are one of the most promising candidates. In addition to security in the future quantum era, these cryptosystems generally have fast encryption and decryption procedures. Code-based cryptography has quite a long history, nearly as old as RSA--one of the best known public key cryptosystems. However, this family of cryptosystems has never been used in practical situations for the reason that it requires large memory for public keys. For instance, Classic McEliece \cite{daniel2020classic} submitted to the NIST PQC project has a public key size of 255 kilobytes for the 128-bit security. To overcome this drawback, a variety of improvements for McEliece's original scheme \cite{mceliece1978public} have been proposed one after another. Generally these improvements can be divided into two categories: one is to substitute Goppa codes used in the McEliece system with other families of codes endowed with special structures, the other is to use codes endowed with the rank metric. However, most of these variants have been shown to be insecure against structural attacks.

The first cryptosystem based on rank metric codes, known as the GPT cryptosystem, was proposed by Gabidulin et al. in \cite{gabidulin1991ideals}. The main advantage of rank-based cryptosystems consists in their compact representation of public keys. For instance, 600 bytes are enough to reach the 100-bit security for the original GPT cryptosystem. After that, applying rank metric codes to the construction of cryptosystems became an important topic in code-based cryptography. Some of the representative variants based on Gabidulin codes can be found in \cite{gabidulin2003reducible,berger2004designing,faure2005new,loidreau2017new,lau2019new}. Unfortunately, most of these variants, including the original GPT cryptosystem, have been completely broken because of Gabidulin codes being highly structrued. Concretely, Gabidulin codes contain a large subspace invariant under the Frobenius transformation, which provides the feasibility for us to distinguish Gabidulin codes from general ones. Based on this observation, various structural attacks \cite{overbeck2008structural,horlemann2018extension,Otmani2018Improved,gaborit2018polynomial,coggia2020security} on the GPT cryptosystem and some of their variants were designed. Apart from Gabidulin codes, another family of rank metric codes, known as the Low Rank Parity Check (LRPC) codes, and a probabilistic encryption scheme based on these codes were proposed in \cite{gaborit2013low,aragon2019low}. Compared to Gabidulin codes, LRPC codes admit a weak algebraic structure. Encryption schemes based on these codes can therefore resist structural attacks designed for Gabidulin codes based cryptosystems. However, this type of cryptosystems generally has a decrypting failure rate, which can be used to devise a reaction attack \cite{samar2019reaction} to recover the private key.

\noindent\textbf{Our contributions} in this paper mainly consist of the following two aspects. 
\begin{itemize}
\item[(1)]We introduce and investigate the so-called semilinear transformations in coding theory. In particular, a semilinear transformation over $\mathbb{F}_{q^m}$ with respect to $\mathbb{F}_q$ is said to be fully linear if it preserves the $\mathbb{F}_{q^m}$-linearity of all linear codes over $\mathbb{F}_{q^m}$. A sufficient and necessary condition for a semilinear transformation being fully linear is given. Furthermore, a fully linear transformation is shown to be a composition of the Frobenius transformation and the stretching transformation. Lastly, we show that a semilinear transformation can be characterized by a linearized permutation polynomial. Meanwhile, we introduce the concept of nonlinearity of a semilinear transformation and prove that the nonlinearity is actually determined by the nonzero coefficients of the associated linearized permutation polynomial.

\item[(2)]We apply the semilinear transformation to construct a public key encryption scheme. Combining Loidreau's technique of disturbing Gabidulin codes, both the public code and its dual in our proposal seem indistinguishable from random codes, which has been verified through extensive experiments in MAGMA. Consequently, all the known distinguisher attacks do not work any longer. To reduce the public key size, we endow the underlying Gabidulin code with the so-called partial cyclic structure. Finally, we obtain a public key cryptosystem with the optimal public key representation compared with some other code-based cryptosystems.
\end{itemize}

The rest of this paper is organized as follows. Section \ref{section2} introduces basic notations used throughout this paper, as well as the definition of Gabidulin codes and partial cyclic codes. Section \ref{section3} presents two hard problems in coding theory and two types of attacks on them that will be useful to estimate the practical security of our proposal. In Section \ref{section4}, we introduce the concept of semilinear transformations, and investigate their algebraic properties when acting on linear codes. Section \ref{section5} is devoted to the description of our new proposal and some notes on the choice of private keys, then we present the security analysis of our proposal in Section \ref{section6}. After that, we suggest parameters for different security levels and make a comparison on public key size with other code-based cryptosystems in Section \ref{section7}. A few concluding remarks will be presented in Section \ref{section8}.

\section{Preliminaries}\label{section2}
We first present notations used throughout this paper, as well as basic concepts of linear codes and rank metric codes. Then we introduce the so-called partial cyclic Gabidulin codes and some related results.
\subsection{Notation and basic concepts}
Let $\mathbb{F}_q$ be a finite field, and $\mathbb{F}_{q^m}$ be an extension field of $\mathbb{F}_q$ of degree $m$. A vector $\bm{a}\in\mathbb{F}_{q^m}^m$ is called a basis vector of $\mathbb{F}_{q^m}$ over $\mathbb{F}_q$ if components of $\bm{a}$ are linearly independent over $\mathbb{F}_q$. Particularly, we call $\alpha$ a polynomial element if $\bm{a}=(1,\alpha,\cdots,\alpha^{m-1})$ forms a basis vector of $\mathbb{F}_{q^m}$ over $\mathbb{F}_q$, and call $\alpha$ a normal element if $\bm{a}=(\alpha,\alpha^q,\cdots,\alpha^{q^{m-1}})\in\mathbb{F}_{q^m}^m$ forms a basis vector. For two positive integers $k$ and $n$, denote by $\mathcal{M}_{k,n}(\mathbb{F}_q)$ the space of all $k\times n$ matrices over $\mathbb{F}_q$, and by $\tl{GL}_n(\mathbb{F}_q)$ the set of all invertible matrices in $\mathcal{M}_{n,n}(\mathbb{F}_q)$. For a matrix $M\in\mathcal{M}_{k,n}(\mathbb{F}_q)$, let $\langle M\rangle_q$ be the vector space spanned by the rows of $M$ over $\mathbb{F}_q$.

An $[n,k]$ linear code $\mathcal{C}$ over $\mathbb{F}_{q^m}$ is a $k$-dimensional subspace of $\mathbb{F}_{q^m}^n$. The dual code of $\mathcal{C}$, denoted by $\mathcal{C}^\perp$, is the orthogonal space of $\mathcal{C}$ under the usual inner product over $\mathbb{F}_{q^m}^n$. A $k\times n$ matrix $G$ over $\mathbb{F}_{q^m}$ of full row rank is called a generator matrix of $\mathcal{C}$ if its row vectors form a basis of $\mathcal{C}$. A generator matrix $H$ of $\mathcal{C}^\perp$ is called a parity-check matrix of $\mathcal{C}$. For a codeword $\bm{c}\in\mathcal{C}$, the Hamming support of $\bm{c}$, denoted by $\tl{Supp}_H(\bm{c})$, is defined to be the set of coordinates of $\bm{c}$ at which the components are nonzero. The Hamming weight of $\bm{c}$, denoted by $\tl{wt}_H(\bm{c})$, is the cardinality of $\tl{Supp}_H(\bm{c})$. The minimum Hamming distance of $\mathcal{C}$ is defined as the minimum Hamming weight of nonzero codewords in $\mathcal{C}$. The rank support of $\bm{c}$, denoted by $\supp(\bm{c})$, is the linear space spanned by the components of $\bm{c}$ over $\mathbb{F}_q$. The rank weight of $\bm{c}$ with respect to $\mathbb{F}_q$, denoted by $\tl{wt}_R(\bm{c})$, is defined to be the dimension of $\supp(\bm{c})$ over $\mathbb{F}_q$. For a matrix $M\in\mathcal{M}_{k,n}(\mathbb{F}_{q^m})$, the rank support $\tl{Supp}_R(M)$ of $M$ is defined to be the linear space spanned by entries of $M$ over $\mathbb{F}_q$. Similarly, the rank weight of $M$ with respect to $\mathbb{F}_q$, denoted by $\tl{wt}_R(M)$, is defined as the dimension of $\tl{Supp}_R(M)$ over $\mathbb{F}_q$.

\subsection{Gabidulin codes}
We first present the concept of Moore matrices and some related results.
\begin{definition}[Moore matrices]
For an integer $i$, we denote by $\alpha^{[i]}=\alpha^{q^i}$ the $i$-th Frobenius power of $\alpha\in\mathbb{F}_{q^m}$. By $\bm{g}^{[i]}$ we denote the component-wise $i$-th Frobenius power of $\bm{g}\in\mathbb{F}_{q^m}^n$. A matrix $G\in\mathcal{M}_{k,n}(\mathbb{F}_{q^m})$ is called a Moore matrix generated by $\bm{g}$ if the $i$-th row vector of $G$ is exactly $\bm{g}^{[i-1]}$ for $1\leqslant i\leqslant k$.
\end{definition}

\begin{remark}
For an $[n,k]$ linear code $\mathcal{C}\subseteq\mathbb{F}_{q^m}^n$, the $i$-th Frobenius power of $\mathcal{C}$ is defined to be $\mathcal{C}^{[i]}=\{\bm{c}^{[i]}:\bm{c}\in\mathcal{C}\}$. Furthermore, it is easy to verify that $\mathcal{C}^{[i]}$ is also an $[n,k]$ linear code over $\mathbb{F}_{q^m}$.
\end{remark}

The following proposition describes simple properties of Moore matrices.
\begin{proposition}
\label{proposition1}
\begin{itemize}
\item[\tl{(1)}] For two Moore matrices $A,B\in\mathcal{M}_{k,n}(\mathbb{F}_{q^m})$, $A+B$ is also a Moore matrix.
\item[\tl{(2)}] For a vector $\bm{a}\in\mathbb{F}_{q^m}^n$ and a matrix $Q\in \mathcal{M}_{n,l}(\mathbb{F}_q)$, let $A\in\mathcal{M}_{k,n}(\mathbb{F}_{q^m})$  be a Moore matrix generated by $\bm{a}$, then $AQ$ is a $k\times l$ Moore matrix generated by $\bm{a}Q$.
\item[\tl{(3)}] For positive integers $k\leqslant n\leqslant m$ and a vector $\bm{a}=(\alpha_1,\cdots,\alpha_n)\in\mathbb{F}_{q^m}^n$ with $\tl{wt}_R(\bm{a})=n$, let $A\in\mathcal{M}_{k,n}(\mathbb{F}_{q^m})$ be a Moore matrix generated by $\bm{a}$. Then $A$ has full row rank,  that is $\rank(A)=k$.
\item[\tl{(4)}] For a vector $\bm{a}\in\mathbb{F}_{q^m}^n$ with $\tl{wt}_R(\bm{a})=s$ where $s\leqslant n$ is a positive integer, let $A\in\mathcal{M}_{k,n}(\mathbb{F}_{q^m})$ be a Moore matrix generated by $\bm{a}$. Then we have $\rank(A)=\min\{s,k\}$.
\end{itemize}
\end{proposition}

Now we introduce the definition of Gabidulin codes.
\begin{definition}[Gabidulin codes]
For positive integers $k\leqslant n\leqslant m$ and $\bm{g}\in\mathbb{F}_{q^m}^n$ with $\tl{wt}_R(\bm{g})=n$. Let $G\in\mathcal{M}_{k,n}(\mathbb{F}_{q^m})$ be a Moore matrix generated by $\bm{g}$, then the $[n,k]$ Gabidulin code generated by $\bm{g}$ is defined to be the linear space $\langle G\rangle _{q^m}$.
\end{definition}

Gabidulin codes can be seen as an analogue of generalized Reed-Solomon (GRS) codes in the rank metric, both of which have pretty good algebraic properties. An $[n,k]$ Gabidulin code has minimum rank distance $d=n-k+1$ \cite{horlemann2015new} and can therefore correct up to $\left\lfloor\frac{n-k}{2}\right\rfloor$ rank errors in theory. Efficient decoding algorithms for Gabidulin codes can be found in \cite{gabidulin1985theory,loidreau2005welch,richter2004error}.

\subsection{Partial cyclic codes}
Lau and Tan \cite{lau2019new} proposed the use of partial cyclic codes to shrink public key size in rank metric based cryptography. Now we introduce this family of codes and present some related results.
\begin{definition}[Partial circulant matrices]
For a vector $\bm{m}\in\mathbb{F}_q^n$, the circulant matrix induced by $\bm{m}$ is a matrix $M\in\mathcal{M}_{n,n}(\mathbb{F}_q)$ whose first row is $\bm{m}$ and $i$-th row is obtained by cyclically right shifting its $i-1$-th row for $2\leqslant i\leqslant n$. For a positive integer $k\leqslant n$, the $k\times n$ partial circulant matrix generated by $\bm{m}$, denoted by $ \tl{PC}_k(\bm{m})$, is defined to be the first $k$ rows of $M$. Particularly, we denote by $ \tl{PC}_n(\bm{m})$ the circulant matrix generated by $\bm{m}$. Furthermore, we denote by $ \tl{PC}_n(\mathbb{F}_q)$ the set of all $n\times n$ circulant matrices over $\mathbb{F}_q$.
\end{definition}

\begin{remark}
Chalkley \cite{chalkley1975circulant} proved that $ \tl{PC}_n(\mathbb{F}_q)$ forms a commutative ring under usual matrix addition and multiplication. Let $\bm{1}=(1,0,\cdots,0)\in\mathbb{F}_q^n$, then $ \tl{PC}_k(\bm{m})= \tl{PC}_k(\bm{1})\cdot \tl{PC}_n(\bm{m})$ for any $\bm{m}\in\mathbb{F}_q^n$. It follows that for a $k\times n$ partial circulant matrix $A$ over $\mathbb{F}_q$ and $B\in \tl{PC}_n(\mathbb{F}_q)$, $AB$ is also a $k\times n$ partial circulant matrix.
\end{remark}

The following two propositions first describe a sufficient and necessary condition for a circulant matrix being invertible, and then make an accurate estimation on the number of invertible circulant matrices over $\mathbb{F}_q$.
\begin{proposition}\cite{otmani2010crypt}
For a vector $\bm{m}=(m_0,\cdots,m_{n-1})\in\mathbb{F}_q^n$, we define $\bm{m}(x)=\sum_{i=0}^{n-1}m_ix^i\in\mathbb{F}_q[x]$. A sufficient and necessary condition for $ \tl{PC}_n(\bm{m})$ being invertible is $\gcd(\bm{m}(x),x^n-1)=1$.
\end{proposition}

\begin{proposition}\cite{mullen2013handbook}\label{proposition2}
For a polynomial $f(x)\in\mathbb{F}_q[x]$ of degree $n$, let $g_1(x),\cdots,g_r(x)\in\mathbb{F}_q[x]$ be $r$ distinct irreducible factors of $f(x)$, i.e. $f(x)=\prod_{i=1}^rg_i(x)^{e_i}$ for some positive integers $e_1,\cdots,e_r$. Let $d_i=\deg(g_i(x))$ for $1\leqslant i\leqslant r$, then 
\begin{align}
\Phi_q(f(x))=q^n\prod_{i=1}^r(1-\frac{1}{q^{d_i}}),
\end{align}
where $\Phi_q(f(x))$ denotes the number of polynomials relatively prime to $f(x)$ of degree less than $n$.
\end{proposition}

Now we introduce the so-called partial cyclic codes.
\begin{definition}[Partial cyclic codes]
For a vector $\bm{a}\in\mathbb{F}_q^n$, let $G= \tl{PC}_k(\bm{a})$ be a partial circulant matrix induced by $\bm{a}$. An $[n,k]$ linear code $\mathcal{C}=\langle G\rangle_q$ is called a partial cyclic code generated by $\bm{a}$.
\end{definition}

\begin{remark}\label{remark1}
For a normal element $g\in\mathbb{F}_{q^n}$ with respect to $\mathbb{F}_q$, let $\bm{g}=(g^{[n-1]},g^{[n-2]},\cdots,g)$ and $G= \tl{PC}_k(\bm{g})$. Easily it can be verified that $G$ is a $k\times n$ Moore matrix. An $[n,k]$ linear code $\mathcal{G}=\langle G\rangle_{q^n}$ is called a partial cyclic Gabidulin code generated by $\bm{g}$.
\end{remark}

As for the total number of $[n,k]$ partial cyclic Gabidulin codes over $\mathbb{F}_{q^n}$, or equivalently the total number of normal elements of $\mathbb{F}_{q^n}$ over $\mathbb{F}_q$, we present the following proposition.
\begin{proposition}\label{proposition4}\cite{mullen2013handbook}
Normal elements of $\mathbb{F}_{q^n}$ over $\mathbb{F}_q$ are in one-to-one correspondence to circulant matrices in $\tl{GL}_n(\mathbb{F}_q)$, which implies that the total number of normal elements of $\mathbb{F}_{q^n}$ over $\mathbb{F}_q$ can be evaluated as $\Phi_q(n)=\Phi_q(x^n-1)$.
\end{proposition}

\section{Hard problems in coding theory}\label{section3}
This section mainly discusses classical hard problems in coding theory on which the security of code-based cryptosystems relies, as well as the best known attacks on them that will be useful to estimate the security level of our proposal later in this paper. 

\begin{definition}[Syndrome Decoding (SD) Problem]
Given positive integers $n,k$ and $t$, let $H$ be an $(n-k)\times n$ matrix over $\mathbb{F}_q$ of full rank and $\bm{s}\in\mathbb{F}_q^{n-k}$. The SD problem with parameters $(q,n,k,t)$ is to search for a vector $\bm{e}\in\mathbb{F}_q^n$ such that $\bm{s}=\bm{e}H^T$ and $\tl{wt}_H(\bm{e})=t$.
\end{definition}

The SD problem, proved to be NP-complete by Berlekamp et al. in \cite{inherent1978}, plays a crucial role in both complexity theory and code-based cryptography. The first instance of the SD problem being used in code-based cryptography is the McEliece cryptosystem \cite{mceliece1978public} based on Goppa codes. A rank metric counterpart of this problem is the rank syndrome decoding problem described as follows.

\begin{definition}[Rank Syndrome Decoding (RSD) Problem]\label{definition}
Given positive integers $m,n,k$ and $t$, let $H$ be an $(n-k)\times n$ matrix over $\mathbb{F}_{q^m}$ of full rank and $\bm{s}\in\mathbb{F}_{q^m}^{n-k}$. The RSD problem with parameters $(q,m,n,k,t)$ is to search for a vector $\bm{e}\in\mathbb{F}_{q^m}^n$ such that $\bm{s}=\bm{e}H^T$ and $\tl{wt}_R(\bm{e})=t$.
\end{definition}

The RSD problem is an important issue in rank metric based cryptography, which has been used for designing cryptosystems since the proposal of the GPT cryptosystem \cite{gabidulin1991ideals} in 1991. However, the hardness of this problem had never been proved until the work in \cite{hardness2016}, where the authors gave a randomized reduction of the SD problem to the RSD problem.

Generally speaking, attacks on the RSD problem can be divided into two categories, namely the combinatorial attack and algebraic attack. The main idea of combinatorial attacks consists in solving a linear system obtained from the parity-check equation, whose unknowns are components of $e_i\,(1\leqslant i\leqslant n)$ with respect to a potential support of $\bm{e}$. Up to now, the best known combinatorial attacks can be found in\cite{technique2002,complexity2016,algorithm2018}, as summarized in Table \ref{table1}.

\begin{table}[!h]\label{table1}
\setlength{\abovecaptionskip}{0.1cm}
\setlength{\belowcaptionskip}{-0.2cm}
\begin{center}
\begin{tabular}{|c|c|}
\hline
Attack & \makecell*[c]{Complexity}\\
\hline
\cite{technique2002} & \makecell*[c]{$\mathcal{O}\left(\min\left\{ m^3t^3q^{(t-1)(k+1)},(k+t)^3t^3q^{(t-1)(m-t)}\right\} \right)$}\\
\hline
\cite{complexity2016} & \makecell*[c]{$\mathcal{O}\left((n-k)^3m^3q^{\min\left\{t\left\lceil\frac{mk}{n}\right\rceil,(t-1)\left\lceil\frac{m(k+1)}{n}\right\rceil\right\}}\right)$}\\
\hline
\cite{algorithm2018} & \makecell*[c]{$\mathcal{O}\left((n-k)^3m^3q^{t\left\lceil \frac{m(k+1)}{n}\right\rceil-m}\right)$}\\
\hline
\end{tabular}
\end{center}
\caption{Best known combinatorial attacks on the RSD problem.}\label{table1}
\end{table}

As for the algebraic attack, the main idea consists in converting an RSD instance into a quadratic system and then solving this system using algebraic approaches. Here in this paper, we mainly consider the attacks proposed in \cite{improvements2020,algebraic2020,complexity2016}, whose complexity and applicable condition are summarized in Table \ref{table2}.

\begin{table}[!h]\label{table2}
\setlength{\abovecaptionskip}{-0.1cm}
\setlength{\belowcaptionskip}{-0.2cm}
\begin{center}
\begin{tabular}{|c|c|c|}
\hline
Attack & Condition & \makecell*[c]{Complexity}\\
\hline
\cite{complexity2016}& $\left\lceil\frac{(t+1)(k+1)-(n+1)}{t}\right\rceil\leqslant k$ & \makecell*[c]{$\mathcal{O}\left( k^3t^3q^{t\left\lceil\frac{(t+1)(k+1)-(n+1)}{t}\right\rceil}\right)$}\\
\hline
\cite{improvements2020} & \multirowcell{3}{$m\binom{n-k-1}{t}\geqslant \binom{n}{t}-1$} & \makecell*[c]{$\mathcal{O}\left(m\binom{n-p-k-1}{t}\binom{n-p}{t}^{\omega-1}\right)$, where $\omega=2.81$ and \\ $p=\min\{1\leqslant i\leqslant n:m\binom{n-i-k-1}{t}\geqslant \binom{n-i}{t}-1\}$}\\
\cline{1-1}\cline{3-3}
\cite{algebraic2020} & & \makecell*[c]{$\mathcal{O}\left(\left(\frac{((m+n)t)^t}{t!}\right)^\omega\right)$}\\
\hline
\cite{improvements2020} & \multirowcell{3}{$m\binom{n-k-1}{t}< \binom{n}{t}-1$} & \makecell*[c]{$\mathcal{O}\left(q^{at}m\binom{n-k-1}{t}\binom{n-a}{t}^{\omega-1}\right)$, where $a=$\\ $\min\{1\leqslant i\leqslant n:m\binom{n-k-1}{t}\geqslant\binom{n-i}{t}-1\}$}\\
\cline{1-1}\cline{3-3}
\cite{algebraic2020} & &\makecell*[c]{$\mathcal{O}\left(\left(\frac{((m+n)t)^{t+1}}{(t+1)!}\right)^\omega\right)$}\\
\hline
\end{tabular}
\end{center}
\caption{Best known algebraic attacks on the RSD problem.}\label{table2}
\end{table}

\section{Semilinear transformations}\label{section4}
In this section, we first introduce the concept of semilinear transformations and investigate their properties when acting on linear codes. Then we discuss these transformations from the perspective of linearized permutation polynomials.
\subsection{Semilinear transformations}
Note that $\mathbb{F}_{q^m}$ can be seen as an $m$-dimensional linear space over $\mathbb{F}_q$. Let $\bm{a}=(\alpha_1,\cdots,\alpha_m)$ and $\bm{b}=(\beta_1,\cdots,\beta_m)$ be two basis vectors of $\mathbb{F}_{q^m}$ over $\mathbb{F}_q$. For any $\alpha=\sum_{i=1}^m\lambda_i\alpha_i\in\mathbb{F}_{q^m}$ with $\lambda_i\in\mathbb{F}_q$, we define an automorphism of $\mathbb{F}_{q^m}$ as
\begin{align*}
\varphi(\alpha)=\sum_{i=1}^m\lambda_i\varphi(\alpha_i)=\sum_{i=1}^m\lambda_i\beta_i.
\end{align*}
It is easy to verify that $\varphi$ is $\mathbb{F}_q$-linear, or equivalently the following property holds
\[\varphi(\lambda_1\alpha+\lambda_2\beta)=\lambda_1\varphi(\alpha)+\lambda_2\varphi(\beta)\]
for any $\alpha,\beta\in\mathbb{F}_{q^m}$ and $\lambda_1,\lambda_2\in\mathbb{F}_q$. By $\aut$ we denote the set of all $\mathbb{F}_q$-linear automorphisms of $\mathbb{F}_{q^m}$.

In what follows, we will do further study on this type of transformations. Firstly, we present a basic fact about the $\mathbb{F}_q$-linear automorphisms of $\mathbb{F}_{q^m}$.
\begin{proposition}\label{proposition6}
The total number of $\mathbb{F}_q$-linear automorphisms of $\mathbb{F}_{q^m}$ is \[|\aut|=\prod_{i=0}^{m-1}(q^m-q^i).\] 
\end{proposition}

\begin{proof}
It is easy to see that all the transformations in $\aut$ are in one-to-one correspondence to $\tl{GL}_m(\mathbb{F}_q)$. By evaluating the cardinality of $\tl{GL}_m(\mathbb{F}_q)$, we obtain the conclusion immediately.
\end{proof}

Let $\varphi\in\aut$ be an $\mathbb{F}_q$-linear automorphism of $\mathbb{F}_{q^m}$. For $\bm{v}=(v_1,\cdots,v_n)\in\mathbb{F}_{q^m}^n$, let $\varphi(\bm{v})=(\varphi(v_1),\cdots,\varphi(v_n))$. For $\mathcal{V}\subseteq\mathbb{F}_{q^m}^n$, let $\varphi(\mathcal{V})=\{\varphi(\bm{v}):\bm{v}\in\mathcal{V}\}$. For $M=(M_{ij})\in\mathcal{M}_{k,n}(\mathbb{F}_{q^m})$, let $\varphi(M)=(\varphi(M_{ij}))$. In these situations, we call $\varphi$ a semilinear transformation over $\mathbb{F}_{q^m}/\mathbb{F}_q$.

For a vector $\bm{c}\in\mathbb{F}_{q^m}^n$, a natural question is how the Hamming (rank) weight of $\bm{c}$ behaves under the action of $\varphi\in\aut$. For this reason, we introduce the following proposition.
\begin{proposition}\label{proposition}
A semilinear transformation over $\mathbb{F}_{q^m}/\mathbb{F}_q$ is isometric in both the Hamming metric and the rank metric.
\end{proposition}

\begin{proof}
For any $\alpha\in\mathbb{F}_{q^m}$ and $\varphi\in\aut$, obviously $\varphi(\alpha)=0$ holds if and only if $\alpha=0$. Hence we have $\tl{Supp}_H(\varphi(\bm{v}))=\tl{Supp}_H(\bm{v})$ for any $\bm{v}\in\mathbb{F}_{q^m}^n$, which implies that $\tl{wt}_H(\varphi(\bm{v}))=\tl{wt}_H(\bm{v})$.

As for the rank metric, let $\bm{v}\in\mathbb{F}_{q^m}^n$ such that $\tl{wt}_R(\bm{v})=n$. If $\tl{wt}_R(\varphi(\bm{v}))<n$, then there exists $\bm{b}\in\mathbb{F}_q^n\backslash \{\bm{0}\}$ such that $\varphi(\bm{v})\bm{b}^T=\varphi(\bm{v}\bm{b}^T)=0$. This implies that $\bm{v}\bm{b}^T=0$, which conflicts with $\tl{wt}_R(\bm{v})=n$. More generally, suppose that $\tl{wt}_R(\bm{v})=r<n$. Then there exist $\bm{v}^*\in\mathbb{F}_{q^m}^r$ with $\tl{wt}_R(\bm{v}^*)=r$ and $Q\in \tl{GL}_n(\mathbb{F}_q)$ such that $\bm{v}=(\bm{v}^*|\bm{0})Q$. It follows that $\varphi(\bm{v})=(\varphi(\bm{v}^*)|\bm{0})Q$ and therefore $\tl{wt}_R(\varphi(\bm{v}))=\tl{wt}_R(\varphi(\bm{v}^*))=r$.
\end{proof}

\begin{remark}
Let $\mathbb{K}$ be an extension field of $\mathbb{F}_{q^m}$. With a similar analysis to Proposition $\ref{proposition}$, we can deduce from a straightforward verification that a semilinear transformation over $\mathbb{K}/\mathbb{F}_{q^m}$ preserves the rank weight of a vector in $\mathbb{K}^n$ with respect to $\mathbb{F}_q$.
\end{remark}

For $\varphi\in\aut$ and a linear code $\mathcal{C}\subseteq\mathbb{F}_{q^m}^n$, it is easy to verify that $\varphi(\mathcal{C})$ is an $\mathbb{F}_q$-linear space, but generally no longer $\mathbb{F}_{q^m}$-linear. Based on this observation, we classify the semilinear transformations over $\mathbb{F}_{q^m}/\mathbb{F}_q$ as follows.

\begin{definition}
Let $\mathcal{C}\subseteq\mathbb{F}_{q^m}^n$ be an $[n,k]$ linear code and $\varphi\in\aut$. If $\varphi(\mathcal{C})$ is also an $\mathbb{F}_{q^m}$-linear code, we say that $\varphi$ is linear over $\mathcal{C}$. Otherwise, we say that $\varphi$ is semilinear over $\mathcal{C}$. If $\varphi$ is linear over all linear codes over $\mathbb{F}_{q^m}$, we say that $\varphi$ is fully linear over $\mathbb{F}_{q^m}$. Otherwise, we say that $\varphi$ is semilinear over $\mathbb{F}_{q^m}$.
\end{definition}

The following theorem provides a sufficient and necessary condition for $\varphi\in\aut$ being fully linear over $\mathbb{F}_{q^m}$.
\begin{theorem}\label{theorem}
Let $\bm{a}=(\alpha_1,\cdots,\alpha_m)$ be a basis vector of $\mathbb{F}_{q^m}$ over $\mathbb{F}_q$ and $\varphi\in\aut$. Let $A=\left[\varphi(\alpha_1\bm{a})^T,\cdots,\varphi(\alpha_m\bm{a})^T\right]^T$, then a sufficient and necessary condition for $\varphi$ being fully linear is $\rank(A)=1$.
\end{theorem}

\begin{proof}
On the necessity aspect. Let $\mathcal{C}=\langle\bm{a}\rangle_{q^m}$ and $\bm{a}_i=\varphi(\alpha_i\bm{a})$ be the $i$-th row vector of $A$. Note that $\varphi$ is fully linear over $\mathbb{F}_{q^m}$, then  $\varphi$ is linear over $\mathcal{C}$, or equivalently $\varphi(\mathcal{C})$ is $\mathbb{F}_{q^m}$-linear. Let $k=\dim_{q^m}(\varphi(\mathcal{C}))$, then $(q^m)^k=|\varphi(\mathcal{C})|=|\mathcal{C}|=q^m$ and therefore $k=1$. The conclusion is proved immediately because of $\bm{a}_i\,(1\leqslant i\leqslant m)$ being contained in $\varphi(\mathcal{C})$.

On the sufficiency aspect. Let $\mathcal{V}=\{\sum_{j=1}^m\lambda_j\bm{a}_j:\lambda_j\in\mathbb{F}_q\}$ and $\mathcal{V}_i=\{\mu\bm{a}_i:\mu\in\mathbb{F}_{q^m}\}$ for any $1\leqslant i\leqslant m$. Note that $A$ has rank $1$ over $\mathbb{F}_{q^m}$, then there exists $\mu_{ij}\in\mathbb{F}_{q^m}^*=\mathbb{F}_{q^m}\backslash\{0\}$ such that $\bm{a}_j=\mu_{ij}\bm{a}_i$ for any $1\leqslant i,j\leqslant m$. It follows that $\mathcal{V}=\{\sum_{j=1}^m\lambda_j\mu_{ij}\bm{a}_i:\lambda_j\in\mathbb{F}_q\}$ and furthermore $\mathcal{V}\subseteq\mathcal{V}_i$. Together with $|\mathcal{V}|=|\mathcal{V}_i|=q^m$, we have $\mathcal{V}=\mathcal{V}_i$. Hence for any $\mu\in\mathbb{F}_{q^m}$, there exist $\lambda_{i1},\cdots,\lambda_{im}\in\mathbb{F}_q$ such that $\mu\bm{a}_i=\sum_{j=1}^m\lambda_{ij}\bm{a}_j$.

Let $\mathcal{C}$ be an arbitrary linear code over $\mathbb{F}_{q^m}$. For any $\bm{c}\in\varphi(\mathcal{C})$, there exists $\bm{u}\in\mathcal{C}$ such that $\bm{c}=\varphi(\bm{u})$. Meanwhile, there exists $M\in\mathcal{M}_{m,n}(\mathbb{F}_q)$ such that $\bm{u}=\bm{a}M$. It follows that
\[ \bm{a}_jM=\varphi(\alpha_j\bm{a})M=\varphi(\alpha_j\bm{a}M)=\varphi(\alpha_j\bm{u})\in\varphi(\mathcal{C})\] 
 for any $1\leqslant j\leqslant m$. Assume that $\sum_{i=1}^ma_i\alpha_i=1$ for $a_i\in\mathbb{F}_q$, then $\bm{u}=\bm{a}M=\sum_{i=1}^ma_i\alpha_i\bm{a}M$. Hence
\begin{align*}
\mu\bm{c}&=\mu\varphi(\bm{u})=\mu\varphi(\sum_{i=1}^ma_i\alpha_i\bm{a}M)=\mu\sum_{i=1}^ma_i\varphi(\alpha_i\bm{a})M=\sum_{i=1}^ma_i\mu\bm{a}_iM.
\end{align*}
Note that for any $\mu\in\mathbb{F}_{q^m}$ and $1\leqslant i\leqslant m$, there exists $\lambda_{ij}\in\mathbb{F}_q$ such that $\mu\bm{a}_i=\sum_{j=1}^m\lambda_{ij}\bm{a}_j$. Hence
\[\mu\bm{c}=\sum_{i=1}^ma_i(\sum_{j=1}^m\lambda_{ij}\bm{a}_j)M=\sum_{i=1}^m\sum_{j=1}^m\lambda_{ij}a_i(\bm{a}_jM)\in\varphi(\mathcal{C})\]
because of $\bm{a}_jM\in\varphi(\mathcal{C})$ and $\varphi(\mathcal{C})$ being $\mathbb{F}_q$-linear. Following this, we conclude that $\varphi(\mathcal{C})$ is $\mathbb{F}_{q^m}$-linear and therefore $\varphi$ is fully linear over $\mathbb{F}_{q^m}$.
\end{proof}

\begin{remark}
Note that $\rank(A)$ is independent of the basis vector. More generally, let $\bm{a}_1$ and $\bm{a}_2$ be another two basis vectors of $\mathbb{F}_{q^m}$ over $\mathbb{F}_q$, then there exist $Q_1,Q_2\in \tl{PC}_n(\mathbb{F}_q)\cap \tl{GL}_n(\mathbb{F}_q)$ such that $\bm{a}_1=\bm{a}Q_1$ and $\bm{a}_2=\bm{a}Q_2$. Let $A'=\varphi(\bm{a}_1^T\bm{a}_2)$, then $A'=\varphi((\bm{a}Q_1)^T\bm{a}Q_2)=\varphi(Q_1^T\bm{a}^T\bm{a}Q_2)=Q_1^TAQ_2$, which implies that $\rank(A)=\rank(A')$.
\end{remark}

To characterize how a semilinear transformation disturbs the algebraic structure of linear codes, we introduce the concept of nonlinearity of a semilinear transformation. Formally, we present the following definition.

\begin{definition}
Let  $\varphi\in\aut$ and $\bm{a}=(\alpha_1,\cdots,\alpha_m)$ be a basis vector of $\mathbb{F}_{q^m}$ over $\mathbb{F}_q$. Let $A\in\mathcal{M}_{m,m}(\mathbb{F}_{q^m})$ be a matrix as defined in Theorem \ref{theorem}. The nonlinearity of $\varphi$ with extension degree $m$ is defined to be $\tl{NL}_m(\varphi)=\frac{r}{m}$ where $r=\rank(A)$.
\end{definition}

\begin{remark}
For $\varphi\in\aut$, it is clear that $\tl{NL}_m(\varphi)$ takes values in $\{\frac{1}{m},\frac{2}{m},\ldots,1\}$.
\end{remark}

The following theorem gives an accurate count of fully linear transformations over $\mathbb{F}_{q^m}/\mathbb{F}_q$.

\begin{theorem}\label{theorem2}
The total number of fully linear transformations in $\aut$ is $m(q^m-1)$.
\end{theorem}

\begin{proof}
Let $\varphi\in\aut$ and $\bm{a}=(1,\alpha,\cdots,\alpha^{m-1})$ where $\alpha$ is a polynomial element of $\mathbb{F}_{q^m}$ over $\mathbb{F}_q$. By Theorem \ref{theorem}, a necessary condition for $\varphi$ being fully linear is that $\varphi(\alpha\bm{a})=\gamma\varphi(\bm{a})$, namely
\begin{align}\label{equation4}
(\varphi(\alpha),\varphi(\alpha^2),\cdots,\varphi(\alpha^m))=\gamma(\varphi(1),\varphi(\alpha),\cdots,\varphi(\alpha^{m-1}))
\end{align}
holds for some $\gamma\in\mathbb{F}_{q^m}$. Assume that $\varphi(1)=\beta\in\mathbb{F}_{q^m}^*$, then we can deduce from (\ref{equation4}) that 
\[\varphi(\alpha^i)=\gamma\varphi(\alpha^{i-1})=\gamma^i\beta \tl{ for } 1\leqslant i\leqslant m.\] 

Let $f(x)=x^m+\sum_{i=0}^{m-1}a_ix^i\in\mathbb{F}_q[x]$ be the minimal polynomial of $\alpha$, then it follows that
\begin{align}\label{equation2}
f(\alpha)=\alpha^m+\sum_{i=0}^{m-1}a_i\alpha^i=0.
\end{align}
Because of $\varphi$ being $\mathbb{F}_q$-linear, applying $\varphi$ to both sides of (\ref{equation2}) leads to the equation
\begin{align*}
\varphi(\alpha^m)+\sum_{i=0}^{m-1}a_i\varphi(\alpha^i)=\gamma^m\beta+\sum_{i=0}^{m-1}a_i\gamma^{i}\beta=0.
\end{align*} 
This implies that $f(\gamma)=0$, then $\gamma=\alpha^{[i]}$ for some $0\leqslant i\leqslant m-1$. 

Conversely, let $\Gamma=\{\alpha^{[i]}\}_{i=0}^{m-1}$, then it is easy to verify that for any duple $(\gamma,\beta)\in\Gamma\times\mathbb{F}_{q^m}^*$, the semilinear transformation over $\mathbb{F}_{q^m}/\mathbb{F}_q$ determined by $\varphi(\alpha^i)=\beta\gamma^i\,(0\leqslant i\leqslant m-1)$ is fully linear over $\mathbb{F}_{q^m}$. Hence all the fully linear transformations over $\mathbb{F}_{q^m}$ are in one-to-one correspondence to the Cartesian product $\Gamma\times\mathbb{F}_{q^m}^*$, which leads to the conclusion immediately.
\end{proof}

\begin{remark}\label{remark}
For a polynomial element $\alpha$ of $\mathbb{F}_{q^m}$ with respect to $\mathbb{F}_q$, let $\Gamma=\{\alpha^{[i]}\}_{i=0}^{m-1}$ be the set of conjugates of $\alpha$. For any $\gamma\in\Gamma$ and $\beta\in\mathbb{F}_{q^m}^*$, the semilinear transformation $\varphi\in\aut$, determined by $\varphi(\alpha^i)=\beta\gamma^i$ for $0\leqslant i\leqslant m-1$, forms a fully linear transformation over $\mathbb{F}_{q^m}$ according to Theorem \ref{theorem2}. Note that $\gamma$ is a conjugate of $\alpha$, then there exists $0\leqslant j\leqslant m-1$ such that $\gamma=\alpha^{[j]}$. For any $\mu=\sum_{i=0}^{m-1}\lambda_i\alpha^i\in\mathbb{F}_{q^m}$ with $\lambda_i\in\mathbb{F}_q$, we have
\begin{align*}
\varphi(\mu)&=\varphi(\sum_{i=0}^{m-1}\lambda_i\alpha^i)=\sum_{i=0}^{m-1}\lambda_i\varphi(\alpha^i)=\sum_{i=0}^{m-1}\lambda_i\beta\gamma^i=\beta\sum_{i=0}^{m-1}\lambda_i(\alpha^{[j]})^i=\beta(\sum_{i=0}^{m-1}\lambda_i\alpha^i)^{[j]}=\beta \mu^{[j]}.
\end{align*} 
This implies that a fully linear transformation over $\mathbb{F}_{q^m}$ can be seen as a composition of the Frobenius transformation and stretching transformation.
\end{remark}

\begin{theorem}\label{theorem1}
For two positive integers $k<n$, let $\mathcal{C}$ be an $[n,k]$ linear code over $\mathbb{F}_{q^m}$. Let $G=[I_k|A]$ be the systematic generator matrix of $\mathcal{C}$, where $I_k$ is the $k\times k$ identity matrix and $A=(A_{ij})\in\mathcal{M}_{k,n-k}(\mathbb{F}_{q^m})$. Let $\mathcal{S}=\{A_{ij}:1\leqslant i\leqslant k,1\leqslant j\leqslant n-k\}$, then we have the following statements.
\begin{itemize}
\item[\tl{(1)}]If $\mathcal{S}\subseteq\mathbb{F}_q$, then any $\varphi\in\aut$ is linear over $\mathcal{C}$. Furthermore, we have $\varphi(\mathcal{C})=\mathcal{C}$;
\item[\tl{(2)}]If there exists $\alpha\in\mathcal{S}$ such that $\alpha$ is a polynomial element of $\mathbb{F}_{q^m}$ over $\mathbb{F}_q$, then any $\varphi\in\aut$ is fully linear if and only if $\varphi$ is linear over $\mathcal{C}$.
\end{itemize}
\end{theorem}

\begin{proof}
\begin{itemize}
\item[(1)] Let $\bm{g}_i$ be the $i$-th row vector of $G$, then $\varphi(\alpha\bm{g}_i)=\varphi(\alpha)\bm{g}_i$ holds for any $\alpha\in\mathbb{F}_{q^m}$. For any $\bm{c}\in\mathcal{C}$, there exists $\bm{\lambda}=(\lambda_1,\cdots,\lambda_k)\in\mathbb{F}_{q^m}^k$ such that $\bm{c}=\bm{\lambda}G$. Then we have
\begin{align*}
\varphi(\bm{c})=\varphi(\bm{\lambda}G)=\varphi(\sum_{i=1}^k\lambda_i\bm{g}_i)=\sum_{i=1}^k\varphi(\lambda_i\bm{g}_i)=\sum_{i=1}^k\varphi(\lambda_i)\bm{g}_i\in\mathcal{C},
\end{align*}
which suggests that $\varphi(\mathcal{C})\subseteq\mathcal{C}$. Together with $|\varphi(\mathcal{C})|=|\mathcal{C}|$, there will be $\varphi(\mathcal{C})=\mathcal{C}$.
\item[(2)] With the necessity being obvious, it suffices to prove the sufficiency. Without loss of generality, we consider the first row vector of $G$ and assume that $\bm{g}_1=(1,0,\cdots,0,\alpha,\star)\in\mathbb{F}_{q^m}^n$, where $\alpha\in\mathbb{F}_{q^m}$ is a polynomial element and ``$\star$" represents some vector in $\mathbb{F}_{q^m}^{n-k-1}$. Note that $\varphi$ is linear over $\mathcal{C}$, or equivalently $\varphi(\mathcal{C})$ is an $\mathbb{F}_{q^m}$-linear code. Apparently $\varphi(\mathcal{C})$ has $\varphi(G)$ as a generator matrix, which implies that there exists $\bm{\lambda}=(\lambda_1,\cdots,\lambda_k)\in\mathbb{F}_{q^m}^k$ such that $\varphi(\beta\bm{g}_1)=\bm{\lambda}\varphi(G)$ for any $\beta\in\mathbb{F}_{q^m}$. It is clear that $\lambda_1\in\mathbb{F}_{q^m}^*$ and $\lambda_i=0$ for $2\leqslant i\leqslant k$, which means $\varphi(\beta\bm{g}_1)$ and $\varphi(\bm{g}_1)$ are linearly dependent over $\mathbb{F}_{q^m}$. Then we can deduce that $(\varphi(\beta),\varphi(\alpha\beta))=\lambda_1(\varphi(1),\varphi(\alpha))$ and furthermore $\varphi(1)\varphi(\alpha\beta)=\varphi(\alpha)\varphi(\beta)$. Let $\gamma=\frac{\varphi(\alpha)}{\varphi(1)}$, then 
\[\varphi(\alpha\beta)=\frac{\varphi(\alpha)}{\varphi(1)}\varphi(\beta)=\gamma\varphi(\beta).\] 
Because of $\alpha$ being a polynomial element, $\bm{a}=(1,\alpha,\cdots,\alpha^{m-1})\in\mathbb{F}_{q^m}^m$ forms a basis vector of $\mathbb{F}_{q^m}$ over $\mathbb{F}_q$. Following this, we have 
\[\varphi(\alpha\bm{a})=(\varphi(\alpha),\cdots,\varphi(\alpha^m))=(\gamma\varphi(1),\cdots,\gamma\varphi(\alpha^{m-1}))=\gamma\varphi(\bm{a}),\] 
and furthermore $\varphi(\alpha^i\bm{a})=\gamma^i\varphi(\bm{a})$ for $0\leqslant i\leqslant m-1$. By Theorem \ref{theorem}, we have that $\varphi$ forms a fully linear transformation over $\mathbb{F}_{q^m}$.
\end{itemize}
\end{proof}

\begin{corollary}
Let $m$ be a prime and $\mathcal{S}$ be defined as above in Theorem \ref{theorem1}. If there exists $\alpha\in\mathcal{S}$ such that $\alpha\notin\mathbb{F}_q$, then any semilinear transformation $\varphi$ over $\mathbb{F}_{q^m}/\mathbb{F}_q$ is fully linear if and only if $\varphi$ is linear over $\mathcal{C}$. 
\end{corollary}

\begin{proof}
Note that $m$ is a prime, then any $\alpha\in\mathbb{F}_{q^m}\backslash\mathbb{F}_q$ is a polynomial element of $\mathbb{F}_{q^m}$ over $\mathbb{F}_q$. Hence the conclusion is proved immediately from Theorem \ref{theorem1}.
\end{proof}

\subsection{Linearized permutation polynomials}\label{section4.2}
Note that an $\mathbb{F}_q$-linear automorphism of $\mathbb{F}_{q^m}$ actually determines a permutation of $\mathbb{F}_{q^m}$, which can be described by a permutation polynomial from Lagrange interpolation. Because of the $\mathbb{F}_q$-linearity, these permutation polynomials are called linearized and have pretty good properties. A linearized polynomial over $\mathbb{F}_{q^m}$ with respect to $\mathbb{F}_q$ is a polynomial of the form $L(x)=\sum_{i=0}^{m-1}a_ix^{[i]}\in\mathbb{F}_{q^m}[x]$. It is easy to verify that, for any $\alpha,\beta\in\mathbb{F}_{q^m}$ and $\lambda\in\mathbb{F}_q$, $L(x)$ has the following properties
\begin{align*}
L(\alpha+\beta)&=L(\alpha)+L(\beta),\\
L(\lambda\alpha)&=\lambda L(\alpha).
\end{align*}

For any $\varphi\in\aut$, we now construct a linearized polynomial $L_\varphi(x)\in\mathbb{F}_{q^m}[x]$ such that $L_\varphi(\alpha)=\varphi(\alpha)$ for any $\alpha\in\mathbb{F}_{q^m}$. Let $\bm{b}=(\beta_1,\ldots,\beta_m)$ be a basis vector of $\mathbb{F}_{q^m}$ over $\mathbb{F}_q$. If $L_\varphi(\beta_i)=\varphi(\beta_i)$ for $1\leqslant i\leqslant m$, then $L_\varphi(\alpha)=\varphi(\alpha)$ for any $\alpha\in\mathbb{F}_{q^m}$. Let $B\in\mathcal{M}_{m,m}(\mathbb{F}_{q^m})$ be a Moore matrix generated by $\bm{b}$, then $(a_0,\ldots,a_{m-1})B=\varphi(\bm{b})$ and hence $(a_0,\ldots,a_{m-1})=\varphi(\bm{b})B^{-1}$. This implies that a semilinear transformation in $\aut$ always leads to a linearized polynomial. But the opposition is not necessarily the case since a linearized polynomial may not induce a permutation of $\mathbb{F}_{q^m}$, such as the Trace function $\tl{Tr}(x)=\sum_{i=0}^{m-1}x^{[i]}$.

For any $\varphi\in\aut$, let $L_\varphi(x)\in\mathbb{F}_{q^m}[x]$ be the linearized permutation polynomial induced by $\varphi$. The following proposition states a fact that one can figure out the nonlinearity of $\varphi$ directly from the number of nonzero coefficients of $L_\varphi(x)$.

\begin{proposition}\label{proposition9}
For any $\varphi\in\aut$, let $L_\varphi(x)=\sum_{i=0}^{m-1}a_ix^{[i]}$ be a linearized permutation polynomial associated to $\varphi$. Let $\bm{a}=(a_0,\ldots,a_{m-1})$, then $\tl{NL}_m(\varphi)=\frac{w}{m}$ where $w=\tl{wt}_H(\bm{a})$.
\end{proposition}
\begin{proof}
Let $w=\tl{wt}_H(\bm{a})$, then there exist $0\leqslant j_0<\cdots<j_{w-1}\leqslant m-1$ such that $a_{j_v}\neq 0$. Let $\bm{b}=(\beta_0,\ldots,\beta_{m-1})$ be a basis vector of $\mathbb{F}_{q^m}$ over $\mathbb{F}_q$, and set
\[B=
\begin{pmatrix}
\beta_0\beta_0&\beta_0\beta_1&\cdots&\beta_0\beta_{m-1}\\
\beta_1\beta_0&\beta_1\beta_1&\cdots&\beta_1\beta_{m-1}\\
\vdots&\vdots&&\vdots\\
\beta_{m-1}\beta_0&\beta_{m-1}\beta_1&\cdots&\beta_{m-1}\beta_{m-1}
\end{pmatrix}=
\begin{pmatrix}
\beta_0\bm{b}\\
\beta_1\bm{b}\\
\vdots\\
\beta_{m-1}\bm{b}
\end{pmatrix}.
\]
Let $A=\varphi(B)$, then 
\[A=L_\varphi(B)=a_{j_0}B^{[j_0]}+a_{j_1}B^{[j_1]}+\cdots+a_{j_{w-1}}B^{[j_{w-1}]}.\] 
It is easy to see that $\rank(A)\leqslant w$ because of $\rank(B^{[j_v]})=1$ for any $0\leqslant v\leqslant w-1$. Let
\[\Lambda=
\begin{pmatrix}
\beta_0^{[j_0]}&\beta_0^{[j_1]}&\cdots&\beta_0^{[j_{w-1}]}\\
\beta_1^{[j_0]}&\beta_1^{[j_1]}&\cdots&\beta_1^{[j_{w-1}]}\\
\vdots&\vdots&&\vdots\\
\beta_{m-1}^{[j_0]}&\beta_{m-1}^{[j_1]}&\cdots&\beta_{m-1}^{[j_{w-1}]}\\
\end{pmatrix}.
\]
It is clear that $\rank(\Lambda)=w$, then there exist $1\leqslant i_0<\cdots<i_{w-1}\leqslant m$ such that the submatrix of $\Lambda$ from the rows indexed by $i_u$ is invertible. Let $I_u=\{i_0,\ldots,i_{w-1}\}$, then by $\Lambda_{I_u}$ we denote the submatrix of $\Lambda$ indexed by $I_u$, and $A_{I_u}$ the submatrix of $A$ respectively. Then 
\[\Lambda_{I_u}^{-1}A_{I_u}=
a_{j_0}\Lambda_{I_u}^{-1}
\begin{pmatrix}
\beta_{i_0}^{[j_0]}\bm{b}^{[j_0]}\\
\beta_{i_1}^{[j_0]}\bm{b}^{[j_0]}\\
\vdots\\
\beta_{i_{w-1}}^{[j_0]}\bm{b}^{[j_0]}\\
\end{pmatrix}+\cdots+a_{j_{w-1}}\Lambda_{I_u}^{-1}
\begin{pmatrix}
\beta_{i_0}^{[j_{w-1}]}\bm{b}^{[j_{w-1}]}\\
\beta_{i_1}^{[j_{w-1}]}\bm{b}^{[j_{w-1}]}\\
\vdots\\
\beta_{i_{w-1}}^{[j_{w-1}]}\bm{b}^{[j_{w-1}]}\\
\end{pmatrix}=
\begin{pmatrix}
a_{j_0}\bm{b}^{[j_0]}\\
a_{j_1}\bm{b}^{[j_1]}\\
\vdots\\
a_{j_{w-1}}\bm{b}^{[j_{w-1}]}\\
\end{pmatrix}.
\]
It follows that $w=\rank(\Lambda_{I_u}^{-1}A_{I_u})\leqslant\rank(A)\leqslant w$, which leads to the conclusion immediately.
\end{proof}

\begin{remark}
It is clear that $\varphi$ is fully linear if and only if the linearized polynomial $L_\varphi(x)$ induced by $\varphi$ has only one nonzero term, i.e. $L_\varphi(x)=ax^{[i]}$ for some $a\in\mathbb{F}_{q^m}^*$ and $0\leqslant i\leqslant m-1$. This accords with the statement in Remark \ref{remark}.
\end{remark}

\section{Our proposal}\label{section5}
In this section, we first give a formal description of our new proposal and then discuss how to choose the private key to avoid some potential structural weakness.

\subsection{Description of our proposal}\label{section5.1}
For a desired security level, choose a field $\mathbb{F}_q$ and positive integers $m,n,k,l,\lambda_1$ and $\lambda_2$ such that $n=lm$. Let $\bm{g}=(g^{[n-1]},g^{[n-2]},\cdots,g)\in\mathbb{F}_{q^n}^n$ be a normal basis vector of $\mathbb{F}_{q^n}$ over $\mathbb{F}_q$, and $G= \tl{PC}_k(\bm{g})\in\mathcal{M}_{k,n}(\mathbb{F}_{q^n})$ a partial circulant matrix generated by $\bm{g}$. Let $\mathcal{G}=\langle G\rangle_{q^n}$ be an $[n,k]$ partial cyclic Gabidulin code. Our proposal consists of the following three procedures.
\begin{itemize}
\item Key generation
\item[]For $i=1,2$, randomly choose an $\mathbb{F}_q$-linear space $\mathcal{V}_i\subseteq\mathbb{F}_{q^n}$ such that $\dim_q(\mathcal{V}_i)=\lambda_i$. Randomly choose $\bm{m}_i\in\mathcal{V}_i^n$ such that $\tl{wt}_R(\bm{m}_i)=\lambda_i$. Let $M_i= \tl{PC}_n(\bm{m}_i)$ and check whether $M_i$ is invertible or not. If not, then rechoose $\bm{m}_i$. Randomly choose a semilinear transformation $\varphi$ over $\mathbb{F}_{q^n}/\mathbb{F}_{q^m}$ such that $\tl{NL}_l(\varphi)\neq \frac{1}{l}$. Let $\bm{g}^*=\varphi(\bm{g}M_1^{-1})M_2^{-1}$, then $\tl{PC}_k(\bm{g}^*)=\varphi(GM_1^{-1})M_2^{-1}$. Let $t=\big\lfloor\frac{n-k}{2\lambda_1\lambda_2}\big\rfloor$, then the public key is published as $(\bm{g}^*,t)$, and the private key is $(\bm{m}_1,\bm{m}_2,\varphi)$.
\item Encryption
\item[]For a plaintext $\bm{x}\in\mathbb{F}_{q^m}^k$, randomly choose $\bm{e}\in\mathbb{F}_{q^n}^n$ with $\tl{wt}_R(\bm{e})=t$. Then the ciphertext corresponding to $\bm{x}$ is computed as 
\[\bm{y}=\bm{x} \tl{PC}_k(\bm{g}^*)+\bm{e}=\bm{x}\varphi(GM_1^{-1})M_2^{-1}+\bm{e}.\]
\item Decryption
\item[]For a ciphertext $\bm{y}\in\mathbb{F}_{q^n}^n$, let $G= \tl{PC}_k(\bm{g}), M_1= \tl{PC}_n(\bm{m}_1)$, and $M_2= \tl{PC}_n(\bm{m}_2)$, then compute 
\[\bm{y}M_2=\bm{x}\varphi(GM_1^{-1})+\bm{e}M_2=\varphi(\bm{x}GM_1^{-1})+\bm{e}M_2,\]  
and 
\[\bm{y}'=\varphi^{-1}(\bm{y}M_2)M_1=\bm{x}G+\varphi^{-1}(\bm{e}M_2)M_1.\] 
Let $\bm{e}'=\varphi^{-1}(\bm{e}M_2)M_1$, then 
\[\tl{wt}_R(\bm{e}')\leqslant \tl{wt}_R(\varphi^{-1}(\bm{e}M_2))\cdot\lambda_1= \tl{wt}_R(\bm{e}M_2)\cdot\lambda_1\leqslant \tl{wt}_R(\bm{e})\cdot\lambda_2\cdot\lambda_1\leqslant \big\lfloor\frac{n-k}{2}\big\rfloor.\] 
Applying the fast decoder of $\mathcal{G}$ to $\bm{y}'$ reveals the error vector $\bm{e}'$, then one can recover $\bm{x}$ by solving the linear system $\bm{x}G=\bm{y}'-\bm{e}'$ with $\mathcal{O}(n^3)$ operations in $\mathbb{F}_{q^n}$.
\end{itemize}

\subsection{A note on the underlying Gabidulin code}
Now we explain why the underlying Gabidulin code is not used as part of the private key. Firstly, we need to introduce the following proposition, which reveals the relationship between two normal basis vectors.

\begin{proposition}\label{proposition7}
Let $\alpha$ be a normal element of $\mathbb{F}_{q^n}$ over $\mathbb{F}_q$, then $\beta\in\mathbb{F}_{q^n}$ is normal if and only if there exists $Q\in \tl{PC}_n(\mathbb{F}_q)\cap \tl{GL}_n(\mathbb{F}_q)$ such that 
\[(\beta^{[n-1]},\beta^{[n-2]},\cdots,\beta)=(\alpha^{[n-1]},\alpha^{[n-2]},\cdots,\alpha)Q.\]
\end{proposition}

\begin{proof}
Trivial from a straightforward verification.
\end{proof}

Note that keeping $G= \tl{PC}_k(\bm{g})$ secret cannot strengthen security of the cryptosystem. Let $\bm{g}'\in\mathbb{F}_{q^n}^n$ be another normal basis vector of $\mathbb{F}_{q^n}$ over $\mathbb{F}_q$. By Proposition \ref{proposition7}, there exists a matrix $Q\in \tl{PC}_n(\mathbb{F}_q)\cap \tl{GL}_n(\mathbb{F}_q)$ such that $\bm{g}=\bm{g}'Q$. Let $G'= \tl{PC}_k(\bm{g}')$, then $G=G'Q$ and 
\[\varphi(GM_1^{-1})M_2^{-1}=\varphi(G'QM_1^{-1})M_2^{-1}=\varphi(G'M_1^{-1})QM_2^{-1}=\varphi(G'M_1^{-1}){M'_2}^{-1},\]
where $M'_2=M_2Q^{-1}\in \tl{PC}_n(\mathbb{F}_{q^n})\cap \tl{GL}_n(\mathbb{F}_{q^n})$ satisfying $\tl{wt}_R(M'_2)=\lambda_2$. Furthermore, it is easy to verify that anyone possessing the knowledge of $\varphi,\bm{g}',M_1$ and $M'_2$ can decrypt any ciphertext in polynomial time. This implies that breaking this cryptosystem can be reduced to recovering $\varphi,M_1$ and $M'_2$. Hence we conclude that it does not make a difference to keep the underlying Gabidulin code secret.

\subsection{On the choice of $\varphi$}
We first explain why the secret transformation $\varphi$ cannot be fully linear over $\mathbb{F}_{q^n}$, then investigate the equivalence between different semilinear transformations. 

\subsubsection{Why $\varphi$ cannot be fully linear}
Assume that $\varphi$ is fully linear over $\mathbb{F}_{q^n}$ with respect to $\mathbb{F}_{q^m}$, then by Remark \ref{remark} there exist $\beta\in\mathbb{F}_{q^n}^*$ and $0\leqslant j\leqslant l-1$ such that
\[\varphi(GM_1^{-1})=\beta (GM_1^{-1})^{[mj]}=\beta G^{[mj]}(M_1^{-1})^{[mj]}=\beta G^{[mj]}(M_1^{[mj]})^{-1}.\]
It follows that
\[\varphi(GM_1^{-1})M_2^{-1}=\beta G^{[mj]}(M_1^{[mj]})^{-1} M_2^{-1}=G^{[mj]}(\beta^{-1}M_2M_1^{[mj]})^{-1}=G'{M'}^{-1},\]
where $G'=G^{[mj]}$ and $M'=\beta^{-1}M_2M_1^{[mj]}$. Note that $\langle G'\rangle_{q^n}$ forms an $[n,k]$ partial cyclic Gabidulin code due to $G'= \tl{PC}_k(\bm{g}^{[mj]})$ and Remark \ref{remark1}. It is clear that $\tl{wt}_R(M')\leqslant\lambda_1\lambda_2$ and one can decrypt any ciphertext with the knowledge of $(G',M')$. This suggests that the cryptosystem degenerates into a sub-instance of Loidreau's cryptosystem \cite{loidreau2017new}, which has been completely broken in some cases \cite{coggia2020security,2020Extending,analysis2021pham}. Hence choosing $\varphi$ to be fully linear over $\mathbb{F}_{q^n}$ is not a smart idea.

Furthermore, we hope that $\varphi(\langle GM_1^{-1}\rangle_{q^n})$ does not preserve the $\mathbb{F}_{q^n}$-linearity. According to our experimental results in MAGMA \cite{bosma1997thema}, the systematic form of $GM_1^{-1}$ always has entries that serve as polynomial elements of $\mathbb{F}_{q^n}$ over $\mathbb{F}_{q^m}$. To generate the secret $\varphi$ that is semilinear over $\langle GM_1^{-1}\rangle_{q^n}$, it sufficies to choose a semilinear transformation over $\mathbb{F}_{q^n}$ due to Theorem \ref{theorem1}.

\subsubsection{Equivalence of semilinear transformations}
For any $\beta\in\mathbb{F}_{q^n}^*$ and a semilinear transformation $\varphi$ over $\mathbb{F}_{q^n}$ with respect to $\mathbb{F}_{q^m}$, it is easy to verify that $\beta\varphi$ is also a semilinear transformation, where $\beta\varphi$ is defined as $\beta\varphi(\alpha)=\beta\cdot\varphi(\alpha)$ for any $\alpha\in\mathbb{F}_{q^n}$. Furthermore, let $\varphi'=\beta\varphi$ and $M'_2=\beta M_2$, then we have $\tl{wt}_R(M'_2)=\tl{wt}_R(M_2)=\lambda_2$ and
\[\varphi(GM_1^{-1})M_2^{-1}=\beta^{-1}\varphi'(GM_1^{-1})M_2^{-1}=\varphi'(GM_1^{-1})(\beta M_2)^{-1}=\varphi'(G{M_1}^{-1}){M'_2}^{-1}.\]
From the perspective of brute-force attack, we say that $\varphi$ and $\varphi'$ are equivalent. We define $\overline{\varphi}=\{\beta\varphi:\beta\in\mathbb{F}_{q^n}^*\}$, called the equivalent class of $\varphi$. For any two transformations $\varphi_1$ and $\varphi_2$, apparently we have either $\overline{\varphi_1}=\overline{\varphi_2}$ or $\overline{\varphi_1}\cap\overline{\varphi_2}=\varnothing$.

Now we make an estimation on the number of nonequivalent semilinear transformations. By Proposition \ref{proposition6}, the number of $\mathbb{F}_{q^m}$-linear automorphisms of $\mathbb{F}_{q^n}$ can be computed as 
\[|\tl{Aut}_{\mathbb{F}_{q^n}/\mathbb{F}_{q^m}}|=\prod_{i=0}^{l-1}(q^n-q^{mi}).\] 
By Theorem \ref{theorem2}, the number of fully linear transformations over $\mathbb{F}_{q^n}$ with respect to $\mathbb{F}_{q^m}$ is $l(q^n-1)$. Denote by $\mathcal{N}(\overline{\varphi})$ the number of nonequivalent semilinear transformations, then
\[\mathcal{N}(\overline{\varphi})=\frac{|\tl{Aut}_{\mathbb{F}_{q^n}/\mathbb{F}_{q^m}}|-l(q^n-1)}{q^n-1}=\prod_{i=1}^{l-1}(q^n-q^{mi})-l\approx q^{(l-1)n}.\]

\subsection{On the choice of $(\bm{m}_1,\bm{m}_2)$}\label{section5.3}
In this subsection, we first investigate how to choose $(\bm{m}_1,\bm{m}_2)$ to avoid some structural weakness, then investigate the equivalence of $\bm{m}_1$'s, and lastly disscuss how to generate $(\bm{m}_1,\bm{m}_2)$ in an efficient way.

\subsubsection{How to choose $(\bm{m}_1,\bm{m}_2)$}
We point out that neither $\bm{m}_1$ or $\bm{m}_2$ should be taken over $\mathbb{F}_{q^m}$, otherwise the proposed cryptosystem will degenerate into a weak scheme. This problem is investigated in the following two cases.

\begin{itemize}
\item[(1)]If $\bm{m}_1\in\mathbb{F}_{q^m}^n$, then there will be $M_1,M_1^{-1}\in \tl{GL}_n(\mathbb{F}_{q^m})$. Following this, we have
\[\varphi(GM_1^{-1})M_2^{-1}=\varphi(G)M_1^{-1}M_2^{-1}=\varphi(G)(M_1M_2)^{-1}=\varphi(G)M^{-1},\]
where $M=M_1M_2$ satisfying $\tl{wt}_R(M)\leqslant \tl{wt}_R(M_1)\cdot \tl{wt}_R(M_2)=\lambda_1\lambda_2$. A straightforward verification shows that if one can recover $\varphi$ and $M$, then one can decrypt any valid ciphertext in polynomial time. Let $G'_{pub}= \tl{PC}_n(\bm{g}^*)$ and $G'= \tl{PC}_n(\bm{g})$, then it can be verified that $G'_{pub}=\varphi(G')M^{-1}$. If one can manage to find $\varphi$, then one can recover $M$ by computing ${G'_{pub}}^{-1}\varphi(G')$. This suggests that breaking this cryptosystem can be reduced to finding the secret $\varphi$.

\item[(2)]If $\bm{m}_2\in\mathbb{F}_{q^m}^n$, then there will be $M_2,M_2^{-1}\in \tl{GL}_n(\mathbb{F}_{q^m})$. Furthermore, we have
\[\varphi(GM_1^{-1})M_2^{-1}=\varphi(GM_1^{-1}M_2^{-1})=\varphi(GM^{-1}),\]
where $M=M_1M_2$ satisfying $\tl{wt}_R(M)\leqslant \tl{wt}_R(M_1)\cdot \tl{wt}_R(M_2)=\lambda_1\lambda_2$. Similarly, a straightforward verification shows that one can decrypt any valid ciphertext with the knowledge of $\varphi,G$ and $M$. If one can manage to find $\varphi$, then one can recover $GM^{-1}$ and then $M$ using a similar method as above. This suggests that breaking this cryptosystem can be reduced to finding the secret $\varphi$.
\end{itemize}

\subsubsection{Equivalence of $\bm{m}_1$'s}\label{subsubsection5.4.2}
For a matrix $Q\in \tl{PC}_n(\mathbb{F}_q)\cap \tl{GL}_n(\mathbb{F}_q)$, let $M'_1=M_1Q$ and $M_2'=M_2Q$, then $\tl{wt}_R(M'_1)=\tl{wt}_R(M_1)$ and $\tl{wt}_R(M'_2)=\tl{wt}_R(M_2)$. Following this, we have
\[\varphi(G{M'_1}^{-1})M_2^{-1}=\varphi(GQ^{-1}M_1^{-1}){M'_2}^{-1}=\varphi(GM_1^{-1})Q^{-1}M_2^{-1}=\varphi(G{M_1}^{-1}){M'_2}^{-1}.\]
From the perspective of a brute-force attack against $\bm{m}_1$, it does not make a difference to multiply $\bm{m}_1$ with a matrix in $\tl{PC}_n(\mathbb{F}_q)\cap \tl{GL}_n(\mathbb{F}_q)$. We define $\overline{\bm{m}}_1=\{\bm{m}_1Q:Q\in \tl{PC}_n(\mathbb{F}_q)\cap \tl{GL}_n(\mathbb{F}_q)\}$, called the equivalent class of $\bm{m}_1$.

Now we make an estimation on the number of nonequivalent $\overline{\bm{m}}_1$'s. For a positive integer $\lambda<n$, let $\mathcal{V}\subseteq\mathbb{F}_{q^n}$ be an $\mathbb{F}_q$-linear space of dimension $\lambda$. For a matrix $M\in \tl{PC}_n(\mathcal{V})\cap \tl{GL}_n(\mathcal{V})$ with $\tl{wt}_R(M)=\lambda$, there exists a decomposition $M=\sum_{j=1}^\lambda\alpha_jA_j$, where $\alpha_j$'s form a basis of $\mathcal{V}$ over $\mathbb{F}_q$ and $A_j$'s are nonzero matrices in $ \tl{PC}_n(\mathbb{F}_q)$. Let $A$ be a matrix whose $j$-th row is exactly the first row of $A_j$, then $A\in\mathcal{M}_{\lambda,n}(\mathbb{F}_q)$ must be of full rank. Denote by $\mathcal{N}(A)$ the number of full-rank matrices in $\mathcal{M}_{\lambda,n}(\mathbb{F}_q)$, and by $\mathcal{N}(\mathcal{V})$ the number of $\lambda$-dimensional $\mathbb{F}_q$-subspaces of $\mathbb{F}_{q^n}$, then
\[\mathcal{N}(A)=\prod_{i=0}^{\lambda-1}(q^n-q^i)\tl{ and }\mathcal{N}(\mathcal{V})=\prod_{j=0}^{\lambda-1}\frac{q^n-q^j}{q^\lambda-q^j}\approx q^{\lambda n-\lambda^2}.\] 
The number of matrices $M\in \tl{PC}_n(\mathbb{F}_{q^n})\cap \tl{GL}_n(\mathbb{F}_{q^n})$ with $\tl{wt}_R(M)=\lambda$ can be evaluated as
\begin{align*}
\mathcal{N}(M)=\mathcal{N}(\mathcal{V})\cdot\mathcal{N}(A)\cdot \xi\approx \xi q^{(2\lambda-1)n-\lambda^2},
\end{align*} 
where $\xi$ denotes the probability of a random $M\in \tl{PC}_n(\mathbb{F}_{q^n})$ with $\tl{wt}_R(M)=\lambda$ being invertible. As for $\xi$, we have the following proposition.
\begin{proposition}\label{proposition8}
If $q^\lambda-q^{\lambda-1}\geqslant 2n$, then $\xi\geqslant \frac{1}{2}$.
\end{proposition}

\begin{proof}
For a $\lambda$-dimensional $\mathbb{F}_q$-linear space $\mathcal{V}\subseteq\mathbb{F}_{q^n}$, denote by $\mathcal{M}_\lambda(\mathcal{V})$ the set of all matrices with rank weight $\lambda$ in $ \tl{PC}_n(\mathcal{V})$. Let $U$ be the set of all singular matrices in $\mathcal{M}_\lambda(\mathcal{V})$, and $V=\mathcal{M}_\lambda(\mathcal{V})\cap \tl{GL}_n(\mathcal{V})$. In what follows, we will construct an injective mapping $\sigma$ from $U$ to $V$. First, we divide $U$ into a certain number of subsets. For a matrix $M\in U$, let $\bm{m}=(m_0,m_1,\cdots,m_{n-1})\in\mathcal{V}^n$ be the first row vector of $M$, namely $M= \tl{PC}_n(\bm{m})$. Let $\overline{M}=\{N\in U: M-N \tl{ is a scalar matrix}\}$, a set of matrices in $U$ that resemble $M$ at the last $n-1$ coordinates. Let $\bm{x}=(x,m_1,\cdots,m_{n-1})$, and $X= \tl{PC}_n(\bm{x})$. Denote by $f(x)\in\mathbb{F}_{q^n}[x]$ the determinant of $X$, then $f(x)$ is a polynomial of degree $n$. In the meanwhile, we have that $|\overline{M}|$ equals the number of roots of $f(x)=0$ in $\mathcal{V}$, which indicates that $|\overline{M}|\leqslant n$. Let $\bm{m}^*=(m_1,\cdots,m_{n-1})$, then it is easy to see that $\tl{wt}_R(\bm{m}^*)\geqslant \lambda-1$. Now we establish the mapping $\sigma$ in the following two cases:
\begin{itemize}
\item[(1)] $\tl{wt}_R(\bm{m}^*)=\lambda-1$. 

For a matrix $M_1\in\overline{M}$, let $\bm{m}_1=(\delta_1,\bm{m}^*)$ be the first row vector of $M_1$. Let $\mathcal{W}=\langle m_1,\cdots,m_{n-1}\rangle_q$, then $\dim_q(\mathcal{W})=\lambda-1$. Because of $q^\lambda-q^{\lambda-1}>n$, there exists $\delta'_1\in\mathcal{V}\backslash \mathcal{W}$ such that $f(\delta'_1)\neq 0$, where $f(x)$ is defined as above. Let $\bm{m}'_1=(\delta'_1,\bm{m}^*)$, then we have $M'_1= \tl{PC}_n(\bm{m}'_1)\in \tl{GL}_n(\mathcal{V})$, and $\tl{wt}_R(\bm{m}'_1)=\lambda$ in the meanwhile. We define $\sigma(M_1)=M'_1$.

For $2\leqslant i\leqslant n$ and a matrix $M_i\in\overline{M}\backslash\{M_j\}_{j=1}^{i-1}$, if any, let $\bm{m}_i=(\delta_i,\bm{m}^*)$ be the first row vector of $M_i$. Because of $q^\lambda-q^{\lambda-1}-(i-1)>n$, there exists $\delta'_i\in\mathcal{V}\backslash (\mathcal{W}\cup\{\delta'_j\}_{j=1}^{i-1})$ such that $f(\delta'_i)\neq 0$. Let $\bm{m}'_i=(\delta'_i,\bm{m}^*)$, then we have $M'_i= \tl{PC}_n(\bm{m}'_i)\in \tl{GL}_n(\mathcal{V})$, and $\tl{wt}_R(\bm{m}'_i)=\lambda$ in the meanwhile. We define $\sigma(M_i)=M'_i$.

\item[(2)] $\tl{wt}_R(\bm{m}^*)=\lambda$. 

For a matrix $M_1\in\overline{M}$, let $\bm{m}_1=(\delta_1,\bm{m}^*)$ be the first row vector of $M_1$. Because of $q^\lambda>n$, there exists $\delta'_1\in\mathcal{V}$ such that $f(\delta'_1)\neq 0$, where $f(x)$ is defined as above. Let $\bm{m}'_1=(\delta'_1,\bm{m}^*)$, then we have $M'_1= \tl{PC}_n(\bm{m}'_1)\in \tl{GL}_n(\mathcal{V})$, and $\tl{wt}_R(\bm{m}'_1)=\lambda$ in the meanwhile. We define $\sigma(M_1)=M'_1$.

For $2\leqslant i\leqslant n$ and a matrix $M_i\in\overline{M}\backslash\{M_j\}_{j=1}^{i-1}$, if any, let $\bm{m}_i=(\delta_i,\bm{m}^*)$ be the first row vector of $M_i$. Because of $q^\lambda-(i-1)>n$, there exists $\delta'_i\in\mathcal{V}\backslash \{\delta'_j\}_{j=1}^{i-1}$ such that $f(\delta'_i)\neq 0$. Let $\bm{m}'_i=(\delta'_i,\bm{m}^*)$, then we have $M'_i= \tl{PC}_n(\bm{m}'_i)\in \tl{GL}_n(\mathcal{V})$, and $\tl{wt}_R(\bm{m}'_i)=\lambda$ in the meanwhile. We define $\sigma(M_i)=M'_i$.
\end{itemize}
It is easy to see that $\sigma$ forms an injective mapping from $U$ to $V$. Apparently $\sigma(U)=\{\sigma(M):M\in U\}\subseteq V$, which implies that $|U|=|\sigma(U)|\leqslant |V|$. Together with $U\cap V=\varnothing$ and $\mathcal{M}_\lambda(\mathcal{V})=U\cup V$, we have that 
\[\xi=\sum_{\substack{\mathcal{V}\subseteq\mathbb{F}_{q^n},\dim_q(\mathcal{V})=\lambda}}|V|\bigg/\sum_{\substack{\mathcal{V}\subseteq\mathbb{F}_{q^n},\dim_q(\mathcal{V})=\lambda}}|\mathcal{M}_\lambda(\mathcal{V})|\geqslant \frac{1}{2}.\]
\end{proof}

\begin{remark}
Proposition \ref{proposition8} provides a sufficient condition for $\xi\geqslant \frac{1}{2}$. Actually, this inequality always holds according to our extensive experiments in MAGMA, even when the sufficient condition is not satisfied. Hence we suppose that $\xi=\frac{1}{2}$ in practice. Finally, the number of nonequivalent $\overline{\bm{m}}_1$'s can be evaluated as 
\[\mathcal{N}(\overline{\bm{m}}_1)=\frac{\mathcal{N}(M_1)}{| \tl{PC}_n(\mathbb{F}_q)\cap \tl{GL}_n(\mathbb{F}_q)|}\approx \frac{q^{2\lambda_1n-\lambda_1^2}}{\Phi_q(n)},\]
where $\Phi_q(n)$ is defined as in Proposition \ref{proposition4}.
\end{remark}

\subsubsection{How to generate $(\bm{m}_1,\bm{m}_2)$ efficiently}
Now we discuss how to generate in an efficient way the secret vector $\bm{m}_i\in \mathbb{F}_{q^n}^n$ such that $\tl{wt}_R(\bm{m}_i)=\lambda_i$ and $\tl{PC}_n(\bm{m}_i)\in GL_n(\mathbb{F}_{q^n})$. Let $M=\sum_{j=1}^\lambda\alpha_jA_j$ be defined as in Section \ref{subsubsection5.4.2}. According to our experiments in MAGMA, if one of these $A_j$'s is chosen to be invertible, then $M$ is invertible with high probability. For example, let $M=\alpha_1A_1+\alpha_2A_2\in \tl{PC}_{20}(\mathbb{F}_{2^{20}})$, where $A_1\neq A_2$ are nonzero matrices in $\tl{PC}_{20}(\mathbb{F}_2)$. If $A_1$ is chosen to be invertible, none of $10000$ $M$'s turn out to be singular. Otherwise, up to $2547$ out of $10000$ $M$'s turn out to be singular. To efficiently generate the secret $\bm{m}_i$, therefore, we adopt the following procedure.
\begin{itemize}
\item[1.] Randomly choose $\alpha_1,\cdots,\alpha_{\lambda_i}\in\mathbb{F}_{q^n}$ linearly independent over $\mathbb{F}_q$;
\item[2.] Randomly choose $\bm{a}_1,\cdots,\bm{a}_{\lambda_i}\in\mathbb{F}_q^n$ linearly independent over $\mathbb{F}_q$ such that $\tl{PC}_n(\bm{a}_1)\in \tl{GL}_n(\mathbb{F}_q)$;
\item[3.] Compute $\bm{m}_i=\sum_{j=1}^{\lambda_i}\alpha_j\bm{a}_j$ and set $M_i=\tl{PC}_n(\bm{m}_i)$.
\item[4.] Check whether or not $M_i$ is invertible. If not, go back to Step 2.
\end{itemize}

\section{Security analysis}\label{section6}
In code-based cryptography, there are mainly two types of attacks on a cryptosystem, namely the structural attack and generic attack. Structural attacks aim to recover the private key from the published information, with which one can decrypt any ciphertext in polynomial time. Generic attacks aim to recover the plaintext directly without the knowledge of the private key. In what follows, we will investigate the security of our new cryptosystem from these two aspects.

\subsection{Structural attacks}\label{section6.1}
Ever since Gabidulin et al. applied Gabidulin codes to construct public key cryptosystems \cite{gabidulin1991ideals}, many variants based on these codes have been proposed. Unfortunately, most of these schemes were completely broken due to the inherent structural vulnerability of Gabidulin codes.

\subsubsection{Overbeck's attack}
The best known structural attacks on McEliece type variants in the rank metric are the one proposed by Overbeck in \cite{overbeck2008structural} and some of its derivations \cite{horlemann2018extension,Otmani2018Improved}. The principle of these attacks consists in an observation that Gabidulin code contains a large subspace invariant under the Frobenius map. To prevent these attacks, Loidreau \cite{loidreau2017new} proposed a new Gabidulin code based cryptosystem, which can be seen as a rank metric counterpart of the BBCRS cryptosystem \cite{Baldi2016Enhanced} based on GRS codes. In Loidreau's cryptosystem, the secret code is disguised by right multiplying a matrix whose inverse is taken over a $\lambda$-dimensional $\mathbb{F}_q$-subspace of $\mathbb{F}_{q^m}$. This method of hiding information about the private key, as claimed by Loidreau, was able to resist the structural attacks mentioned above. A similar technique is applied in our proposal, which we believe can as well prevent these attacks.

\subsubsection{Coggia-Couvreur attack}\label{section6.1.2}
In \cite{coggia2020security}, Coggia and Couvreur proposed an effective method to distinguish the public code of Loidreau's cryptosystem from general ones, and gave a practical key-recovery attack in the case of $\lambda=2$ and the code rate being greater than $1/2$. Instead of operating the public code directly, Coggia and Couvreur considered the dual of the public code. Specifically, let $G_{pub}=GM^{-1}$ be the public matrix of Loidreau's cryptosystem, where $G$ is a generator matrix of an $[n,k]$ Gabidulin code $\mathcal{G}$ over $\mathbb{F}_{q^m}$ and entries of $M$ are contained in a $\lambda$-dimensional $\mathbb{F}_q$-subspace of $\mathbb{F}_{q^m}$. Denote by $H$ a parity-check matrix of $\mathcal{G}$, then $H_{pub}=HM^T$ forms a parity-check matrix of the public code $\mathcal{G}_{pub}=\langle G_{pub}\rangle_{q^m}$. As for the dual code $\mathcal{G}_{pub}^\perp=\langle H_{pub}\rangle_{q^m}$, the Coggia-Couvreur distinguisher states that the following equality holds with high probability
\[\dim_{q^m}(\mathcal{G}_{pub}^\perp+{\mathcal{G}_{pub}^\perp}^{[1]}+\cdots+{\mathcal{G}_{pub}^\perp}^{[\lambda]})=\min\{n,\lambda(n-k)+\lambda\}.\]
For an $[n,k]$ random linear code $\mathcal{C}_{rand}$ over $\mathbb{F}_{q^m}$, however, the following equality holds with high probability
\[\dim_{q^m}(\mathcal{C}_{rand}^\perp+{\mathcal{C}_{rand}^\perp}^{[1]}+\cdots+{\mathcal{C}_{rand}^\perp}^{[\lambda]})=\min\{n,(\lambda+1)(n-k)\}.\] 
Not long after Coggia and Couvreur's work, this attack was generalized to the case of $\lambda=3$ by Ghatak \cite{2020Extending} and then by Pham and Loidreau \cite{analysis2021pham}.

In our proposal, the public matrix is $G_{pub}=\varphi(GM_1^{-1})M_2^{-1}$. To make it easier, we take the parameter $l=2$ as an example. Because of $\varphi$ being semilinear over $\mathbb{F}_{q^n}$, there exists a linearized permutation polynomial $f(x)=\gamma_1x+\gamma_2x^{[m]}$ with $\gamma_i\in\mathbb{F}_{q^n}^*$ such that $\varphi(\alpha)=f(\alpha)$ for any $\alpha\in\mathbb{F}_{q^n}$. Then 
\[G_{pub}=\varphi(GM_1^{-1})M_2^{-1}=(\gamma_1GM_1^{-1}+\gamma_2G^{[m]}(M_1^{-1})^{[m]})M_2^{-1}.\]
On account of the structure of $G$, there exists $Q\in\tl{PC}_n(\mathbb{F}_q)\cap \tl{GL}_n(\mathbb{F}_q)$ such that $G^{[m]}=GQ$, then $G_{pub}=G(\gamma_1M_1^{-1}+\gamma_2Q(M_1^{-1})^{[m]})M_2^{-1}$. Let $G_{pub}=GM^{-1}$, where $M=M'_1M_2$ and
\begin{align*}
M'_1=(\gamma_1M_1^{-1}+\gamma_2Q(M_1^{-1})^{[m]})^{-1}=M_1M_1^{[m]}(\gamma_2QM_1+\gamma_1M_1^{[m]})^{-1}.
\end{align*}
Notice that $\bm{g}$ is public, one can recover $M$ directly by computing $\tl{PC}_n(\bm{g})^{-1}\tl{PC}_n(\bm{g}^*)$. However, it does not mean one can decrypt a ciphertext with the knowledge of $G$ and $M$. This is because $M$ seems quite random and the value of $\tl{wt}_R(M)$ can be very large. For instance, we have run 1000 random tests for $q=2,m=30,n=60$ and $\lambda_1=\lambda_2=2$. It turned out that $\tl{wt}_R(M)\geqslant 54$ for all these tests and the values of $\tl{wt}_R(M)$ occuring most frequently are $58$ for $233$ times, $59$ for $493$ times and $60$ for $210$ times respectively. Consequently, the dual of the public code $\mathcal{G}_{pub}=\langle G_{pub}\rangle_{q^n}$ turns out to be indistinguishable from random codes. To be exact, the following equality holds with high probability according to our experimental results,
\[\dim_{q^m}(\mathcal{G}_{pub}^\perp+{\mathcal{G}_{pub}^\perp}^{[1]}+\cdots+{\mathcal{G}_{pub}^\perp}^{[\lambda]})=\min\{n,(\lambda+1)(n-k)\}.\]
This convinces us that our proposal can prevent Coggia-Couvreur attack.

\subsubsection{Loidreau's attack}
In a talk \cite{loidreau2021analysing} at CBCrypto 2021, Loidreau proposed an attack to recover a decoder of the public code in Loidreau's cryptosystem with a complexity of $\mathcal{O}(((\lambda n+(n-k)^2)m)^\omega q^{(\lambda-1)m})$. With this decoder one can decrypt any ciphertext in polynomial time. Similar to Coggia-Couvreur attack, this attack also considers the dual of the public code. However, an applicable condition for this attack is that the public matrix can be decomposed as $G_{pub}=GM^{-1}$, where $G$ is a generator matrix of a Gabidulin code or one of its subcodes and entries of $M$ are contained in a small $\mathbb{F}_q$-subspace of $\mathbb{F}_{q^n}$. Obviously the public matrix in our proposal does not satisfy this condition according to the analysis in Section \ref{section6.1.2}, which implies that this attack does not work on our new cryptosystem.

\subsubsection{A brute-force attack}\label{section6.1.4}
Now we consider a potential brute-force attack against the duple $(\overline{\varphi},\overline{\bm{m}}_1)$. Notice that for any $\varphi'\in\overline{\varphi}$ and $\bm{m}'_1\in\overline{\bm{m}}_1$, there exists $\bm{m}'_2\in\mathbb{F}_{q^n}^n$ with $\tl{wt}_R(\bm{m}'_2)=\lambda_2$ such that $G_{pub}=\varphi(G{M_1}^{-1})M_2^{-1}=\varphi'(G{M'_1}^{-1}){M'_2}^{-1}$, where $M'_1= \tl{PC}_n(\bm{m}'_1)$ and $M'_2= \tl{PC}_n(\bm{m}'_2)$. Let $G'_{pub}= \tl{PC}_n(\bm{g}^*)$ and $G'= \tl{PC}_n(\bm{g})$, then 
\[G'_{pub}=\varphi(G'M_1^{-1})M_2^{-1}=\varphi'(G'{M'_1}^{-1}){M'_2}^{-1}.\]
This implies that one can compute $M'_2={G'_{pub}}^{-1}\varphi'(G'{M'_1}^{-1})$. Furthermore, a straightforward verification shows that one can decrypt any ciphertext with the knowledge of $\varphi',\bm{m}'_1, \bm{m}'_2$ and the public $\bm{g}$. Apparently the complexity of this brute-force attack by exhausting $(\overline{\varphi},\overline{\bm{m}}_1)$ is $\mathcal{O}(\mathcal{N}(\overline{\varphi})\cdot\mathcal{N}(\overline{\bm{m}}_1))$.

\subsection{Generic attacks}
A legitimate message receiver can always recover the plaintext in polynomial time, while an adversary without the private key has to deal with the underlying RSD problem presented in Section \ref{section3}. Attacks that aim to recover the plaintext directly by solving the RSD problem are called generic attacks, the complexity of which only relates to the parameters of the cryptosystem. In what follows, we will show how to establish a connection between our proposal and the RSD problem.

Let $G_{pub}=\varphi(GM_1^{-1})M_2^{-1}\in\mathcal{M}_{k,n}(\mathbb{F}_{q^n})$ be the public matrix, and $H_{pub}\in\mathcal{M}_{n-k,n}(\mathbb{F}_{q^n})$ a parity-check matrix of the public code $\mathcal{G}_{pub}=\langle G_{pub}\rangle_{q^n}$. Let $\bm{y}=\bm{x}G_{pub}+\bm{e}$ be the received ciphertext, then the syndrome of $\bm{y}$ with respect to $H_{pub}$ can be computed as $\bm{s}=\bm{y}H_{pub}^T=\bm{e}H_{pub}^T$. By Definition \ref{definition}, one obtains an RSD instance of parameters $(q,n,n,k,t)$. Solving this RSD instance by the combinatorial attacks listed in Table \ref{table1} or the algebraic attacks listed in Table \ref{table2} will lead to the error vector $\bm{e}$, then one can recover the plaintext by solving the linear system $\bm{y}-\bm{e}=\bm{x}G_{pub}$.

\section{Parameters and public key sizes}\label{section7}
In this section, we consider the practical security of our proposal against the generic attacks presented in Section \ref{section3}, as well as a brute-force attack against the duple $(\overline{\varphi},\overline{\bm{m}}_1)$ as described in Section \ref{section6.1.4}, with a complexity of $\mathcal{O}(\mathcal{N}(\overline{\varphi})\cdot\mathcal{N}(\overline{\bm{m}}_1))$. The public key in our proposal is a vector in $\mathbb{F}_{q^n}^n$, leading to a public key size of $n^2\cdot \log_2(q)$ bits.  In Table \ref{table3}, we give some suggested parameters for the security level of at least 128 bits, 192 bits, and 256 bits. After that, we make a comparison on public key size with some other code-based cryptosystems in Table \ref{table4}. It is easy to see that our proposal has an obvious advantage over other variants in public key representation.

\begin{table}[!h]
\setlength{\abovecaptionskip}{-0.1cm}
\setlength{\belowcaptionskip}{-0.2cm}
\begin{center}
\begin{tabular}{|ccccccc|c|c|}
\hline
\multicolumn{7}{|c|}{\makecell*[c]{Parameters}} & \multirowcell{2}{\makecell*[c]{Public Key Size}} & \multirowcell{2}{\makecell*[c]{Security}}\\
\cline{1-7}
\makecell*[c]{$q$} & $m$ & $n$ & $k$ & $l$ & $\lambda_1$ & $\lambda_2$ &&\\
\hline
2&55&110&54&2&2&2& 1513 &\makecell*[c]{138}\\
\hline
2&60&120&64&2&2&2& 1800 &\makecell*[c]{197}\\
\hline
2&72&144&72&2&2&2& 2592 &\makecell*[c]{257}\\
\hline
\end{tabular}
\end{center}
\caption{Parameters and public key size (in bytes).}\label{table3}
\end{table}

\begin{table}[!h]
\setlength{\abovecaptionskip}{-0.1cm}
\setlength{\belowcaptionskip}{-0.2cm}
\begin{center}
\begin{tabular}{|l|r|r|r|}
\hline
\backslashbox{Instance}{Security} & \makecell*[c]{128} & \makecell*[c]{192} & \makecell*[c]{256}\\
\hline\rule{0pt}{10pt}
\makecell*[c]{Classic McEliece \cite{daniel2020classic}}& 261120 & 524160 & 1044992\\
\hline\rule{0pt}{10pt}
\makecell*[c]{NTS-KEM \cite{ntskem2019}}& 319488 & 929760 & 1419704\\
\hline\rule{0pt}{10pt}
\makecell*[c]{Guo-Fu II \cite{guo2021expanded}}& 79358 & 212768 & 393422\\
\hline\rule{0pt}{10pt}
\makecell*[c]{Guo-Fu I \cite{guo2021expanded}}& 8993 & 18359 & 37583\\
\hline\rule{0pt}{10pt}
\makecell*[c]{HQC \cite{aguilar2020}}& 2249 & 4522 & 7245\\
\hline\rule{0pt}{10pt}
\makecell*[c]{BIKE \cite{aragon2020}}& 1540 & 3082 & 5121\\
\hline\rule{0pt}{10pt}
\makecell*[c]{Lau-Tan \cite{lau2019new}}&2421&3283& 4409\\
\hline\rule{0pt}{10pt}
\makecell*[c]{Our proposal} & 1513 & 1800 & 2592\\
\hline
\end{tabular}
\end{center}
\caption{Comparison on public key size (in bytes).}\label{table4}
\end{table}

\section{Conclusion}\label{section8}
We have introduced a novel transformation in coding theory, which is defined as an $\mathbb{F}_q$-linear automorphism of $\mathbb{F}_{q^m}$. According to their properites when acting on linear codes over $\mathbb{F}_{q^m}$, these transformations are divided into two categories, namely the fully linear transformation and the semilinear transformation. As an application of semilinear transformations, a new technique is developed to conceal the secret information in code-based cryptosystems. To obtain a small public key size, we exploit the so-called partial cyclic Gabidulin code to construct an encryption scheme, whose security does not rely on the confidentiality of the underlying Gabidulin code. According to our analysis, both the public code and its dual seem indistinguishable from random codes and can therefore resist the existing structural attacks. Furthermore, the system also admits a much smaller public key size compared to some other code-based cryptosystems. For instance, 2592 bytes are enough for our proposal to achieve the security of 256 bits, 403 times smaller than that of Classic McEliece moving onto the third round of the NIST PQC standardization process

\begin{acknowledgements}
This research is supported by the National Key Research and Development Program of China (Grant No. 2018YFA0704703), the National Natural Science Foundation of China (Grant No. 61971243), the Natural Science Foundation of Tianjin (20JCZDJC00610), and the Fundamental Research Funds for the Central Universities of China (Nankai University).
\end{acknowledgements}

\end{document}